\documentclass[conference]{IEEEtran}
\usepackage{cite}
\usepackage[usenames,dvipsnames]{xcolor}
\usepackage[pdftex]{graphicx}
\usepackage{circuitikz}
\usepackage[hidelinks,bookmarks=false]{hyperref}
\usepackage{booktabs,colortbl}
\usepackage{multirow}
\usepackage[cmex10]{amsmath}
\usepackage{mathrsfs}
\usepackage{bbold}
\usepackage{array}
\usepackage{multirow}
\usepackage{booktabs,colortbl}
\usepackage{amssymb}
\usepackage{siunitx}
\usepackage{amssymb}
\usepackage{cases}
\usepackage{algorithm}
\usepackage{algpseudocode}
\usepackage{color,soul}
\usepackage{tikz}
\usetikzlibrary{decorations.pathreplacing}
\usetikzlibrary{arrows,shapes,positioning}
\usetikzlibrary{patterns}
\usepackage{pgfplots}
\pgfplotsset{compat=newest}
\pgfplotsset{plot coordinates/math parser=false}
\newlength\figureheight
\newlength\matlabfigurewidth
\usepackage{multicol}
\usepackage{breqn}
\usepackage{cuted}
\usepackage{amsthm}
\newtheorem{theorem}{Theorem}[section]

\theoremstyle{definition}

\theoremstyle{remark}

\usepackage{mathtools}
\newcommand{\argmin}[1]{\underset{#1}{\text{argmin}}\,}

\newcommand{\minimize}[1]{\underset{#1}{\text{min}}\;}

\newcommand\numberthis{\addtocounter{equation}{1}\tag{\theequation}}

\ifCLASSINFOpdf
\else
\fi
\usepackage{amsmath}
\IEEEoverridecommandlockouts

\begin{document}
%
\title{A rational decentralized generalized Nash equilibrium seeking for energy markets\\
\thanks{The authors would like to thank Innosuisse - Swiss Innovation Agency (CH) and SCCER-FURIES - Swiss Competence Center for Energy Research - Future Swiss Electrical Infrastructure for their financial and technical support to the research work presented in this paper. This work has been sponsored by the Swiss Federal Office of Energy (Project
	nr. SI/501499)}}



\author{\IEEEauthorblockN{Lorenzo Nespoli}
\IEEEauthorblockA{\textit{EPFL}\\Lausanne, Switzerland\\
	\textit{ISAAC, SUPSI}\\ Campus Trevano, Canobbio, Switzerland\\
Email: lorenzo.nespoli@epfl.ch}
\and
\IEEEauthorblockN{Matteo Salani}
\IEEEauthorblockA{\textit{IDSIA, USI/SUPSI} \\ Galleria 2, Manno, Switzerland\\
	Email: matteo.salani@idsia.ch}
\and
\IEEEauthorblockN{Vasco Medici}
\IEEEauthorblockA{\textit{ISAAC, SUPSI}\\ Campus Trevano, Canobbio, Switzerland\\
	Email: vasco.medici@supsi.ch}}


%


\maketitle

\begin{abstract}
We propose a method to design a decentralized energy market which guarantees individual rationality (IR) in expectation, in the presence of system-level grid constraints. We formulate the market as a welfare maximization problem subject to IR constraints, and we make use of Lagrangian duality to model the problem as a n-person non-cooperative game with a unique generalized Nash equilibrium (GNE). We provide a distributed algorithm which converges to the GNE. The convergence and properties of the algorithm are investigated by means of numerical simulations. 
\end{abstract}

%
\IEEEpeerreviewmaketitle

\section{Introduction}
\subsection{Motivations}
The scientific community agrees that in the future the intelligent activation of demand response (DR) will contribute to a reliable power system and price stability on power markets. 
Actuation of DR requires to solve an optimization problem in order to maximize an economic objective, which typically results in a welfare maximization problem (WMP), in which the unweighted sum of the economic costs of a group of agents is minimized. A very similar, and perhaps more studied problem, is the optimal power flow (OPF) problem. The OPF is usually solved in a centralized way by an independent system operator (ISO), in order to minimize the generation cost of a group of distributed power plants, over the set of underlying grid constraints.
When the number of generators increases, the problem could become computationally expensive. Furthermore, retrieving all the generator-specific parameters could become impractical for the ISOs. For these reasons, different decentralized formulations of the OPF exist \cite{Molzahn2017}, which can speed up the computation exploiting parallelization among the different units. Furthermore, solving the problem in a decentralized way allows the generators to keep most of their information and parameters private, increasing privacy and lowering cyber-security concerns. 
The main difference between the OPF and DR setting, is that the second one involves the participation of self-serving agents, which cannot be a-priori trusted by the ISOs. 
This implies that if an agent find it profitable (in terms of its own economic utility), he will compute a different optimization problem from the one provided by the ISO. For this reason, some aspects of DR formulations are better described through a game theoretic framework.  
\subsection{Background and previous work}

In this setting, we must consider that agents can adopt a strategy $s_i(\theta_i)$, which can be in general different from the one suggested by the ISO,  based on their private information (or type), denoted as $\theta_i$, and their belief about the strategy of the other prosumers.
The well-known Vickrey-Clarke-Groves (VCG) mechanism \cite{Makowski1987,Clarke1971,Vickrey1961} belongs to the strategy-proof class of mechanisms and presents other useful theoretical properties, among which being weakly budget-balanced. Anyway, to achieve this, it requires a value redistribution among agents under the form of monetary taxation, such that the tax which applies to agent $i$ is directly or indirectly independent from its actions. This implies that $N$ optimization problems must be solved, each of which is performed without considering a given agent. This makes the computational cost quadratic in $N$. Furthermore, 
VCGs are typically centralized and as such, they do not preserve the privacy of the agents. For example, in \cite{Poolla2017}, a VCG mechanism for virtual inertia is considered, in which bidders send their bidding curves to a center, which solves $N$ independent optimization problems. Since the VCG mechanism guarantees that the best bidding strategy is bidding truthfully, they send their true cost curves $c_i(x_i,\lambda_i)$ to the center. Note anyway that, if the agent's system presents some constraints, $c_i(x_i,\lambda_i)$ must represent them. 
This means that the center must know all the agent constraint sets $\mathcal{X}_i$ in order to solve the VCG. The unfavorable computational cost makes the VCG impractical for combinatorial auctions \cite{Conitzer2006} and problems with a large number of users with a nontrivial objective function. Despite this and other aspects which make it impractical in some cases \cite{Rothkopf2007}, VCGs have been extensively studied since they are the only general purpose incentive compatible mechanisms which maximize social welfare \cite{Tardos2007}.  
In order to preserve agent's privacy, it is possible to retrieve a distributed formulation of VCG using primal-dual decomposition algorithms. Note that distributing the mechanism aggravates the scalability problem of VCG, since the overall computation must now take into account communication delays. A second effect of adopting a decentralized formulation is that we cannot guarantee strategyproofness anymore. This is known as the cost of decentralization \cite{Petcu2008}, which leads to a weaker notion of incentive compatibility, namely ex-post Nash equilibrium (EPNE). Although weaker than a dominant-strategy equilibrium, ex-post Nash equilibrium does not require agents to model the strategies nor types of other agents through belief functions, as it's done using Bayes-Nash equilibrium \cite{Narahari2014}. 
Following this concept, in \cite{Parkes2004} guidelines for distributed implementations of VCG mechanisms are derived. In \cite{Petcu2008} a distributed VCG mechanism which reuses part of the computation done in each subproblem is presented. More recently \cite{Tanaka2018} has proposed a distributed VCG implementation based on dual decomposition, and applied the concept of multistage mechanism, in which different mechanisms are applied at each primal dual update. Also in this case, the proposed algorithm scales quadratically with the number of agents. Another field of research, started with the seminal work of Rosen on n-person non-cooperative games \cite{Rosen1965}, adopt non-VCG mechanisms to reach EPNE \cite{Kim2013a,Gharesifard2016}. This involves allowing a loss in terms of efficiency \cite{Yang2010}, with the benefit of better scalability with respect to the number of agents.  

In this paper we propose a method to guarantee participation constraint, also known as individual rationality (IR): all the prosumers must have a positive return participating in the proposed energy market, with respect to the base case. We ensure IR allowing a coordinator to limit the Lagrangian multipliers associated to the coupling constraints. The rest of the paper is structured as follows: in \ref{s:prob} the specific problem we address is formulated and we show that its associate game mapping is monotone, which is a condition for the uniqueness of the VGNE; in \ref{s:algo} we propose a new algorithm for reaching the GNE, based on the alternating direction method of multipliers (ADMM); in \ref{s:ana} we compare the convergence of the aforementioned algorithm with a recently proposed \cite{Belgioioso2018} preconditioned forward backward (pFB) algorithm for distributed Nash equilibrium seeking.  


\section{Problem formulation}\label{s:prob}

In this work we are interested in a more general problem with respect of the OPF. In particular, we consider the case in which a group of agents which produce and/or consume energy (prosumers now on) can sell their aggregated flexibility to third parties, for example to DSO through demand response programs or to balance responsible parties. The mathematical formulation of this problem is known as the sharing problem:

\begin{equation}\label{eq:sharing}
\begin{aligned}
\argmin{x \in \mathcal{X}}  & e(x) + \sum_i^{N} c_i(x_i)\\
s.t. & \quad A x \leq b
\end{aligned}
\end{equation}

where $\mathcal{X} = \prod_{i=1}^N \mathcal{X}_i$ is the Cartesian product of the prosumers feasible sets, $e(x)$ is a system level objective, $c_i(x_i)$ are the costs of each prosumers and the linear constraints are affine coupling constraints between the prosumers and $x = [x_1^T,..x_N^T] = [x_i]_{i=1}^N $ is the vector of the concatenated actions of all the prosumers. 
Here the affine coupling constraints encode grid constraints, limiting voltage and power in a subset of selected nodes of the grid in which the agents are located. This is possible taking into account the linearized formulation of the power flow equations \cite{Molzahn2017,Almasalma2017}, whose coefficients can be estimated using phasor measurement units \cite{Mugnier2016}, even using smart meter data \cite{Weckx2015}.
The advantage of considering coupling constraints instead of agent-level constraints on voltage and power is given by the fact that the first approach can reach better solutions in terms of total welfare. 

As anticipated in the introduction, we are interested in decomposing problem \eqref{eq:sharing} among the self interested prosumers, in such a way that the induced game presents only one variational GNE, and in the algorithms leading to such an equilibrium. Being the equilibrium unique, rational agents will converge to the EPGNE. This is equivalent to assume that the agents believe their own influence on the prices broadcasted by the sequence of mechanism proposed by the algorithm are negligible, i.e. they are price takers. 

A reasonable way to turn the centralized problem \eqref{eq:sharing} into a non-cooperative game, is to reward each prosumer with a part of the system level objective $e(x)$, based on the amount of energy he produces or consumes during a give period of time:
\begin{equation}\label{eq:utilities}
v(x_i,x_{-i}) = c_i(x_i) + \frac{\vert x_i\vert}{\sum_{i=1}^N \vert x_i \vert} e(x)
\end{equation}

Anyway, this would result in a non-linear and non-convex game. As a first approximation we can replace this repartition rule with fixed (during each horizon) coefficients, based on a moving average:
\begin{equation}
v(x_i,x_{-i}) = c_i(x_i) + \alpha_i e(x)
\end{equation}
where 
\begin{equation}
\alpha_i = \frac{\sum_{k=t-\tau}^t\vert x_{i,k}\vert}{\sum_{k=t-\tau}^t\sum_{i=1}^N \vert x_{i,k} \vert}
\end{equation}
Note that the game $\mathcal{G}(s_i(x),v_i(x))$ induced by the value functions in \eqref{eq:utilities} defines an aggregative game \cite{Jensen2010}, in which the each prosumer influence other's prosumers value only by means of the aggregated actions. 
The induced game can be described as the set of optimization problems \eqref{GNE} in which each prosumer minimizes its own value function $v(x_i, x_{-i})$ and associated KKT conditions \eqref{GNEKKT}.

\begin{equation}\label{GNE}
\begin{cases}
\minimize{x_i \in \mathcal{X}_i} v(x_i, x_{-i})  \\
s.t  \quad Ax\leq b 
\end{cases} \forall i \in N
\end{equation}

\begin{equation}\label{GNEKKT}
KKT(i) = 
\begin{cases}
0\in \partial_{x_i} v_i(x_i,\mathrm{x}_{-i}) + \mathrm{N}_{\mathcal{X}_i} + A_i^T\lambda_i   \\
0 \leq \lambda_i \perp -(Ax-b) \geq 0 
\end{cases} 
\end{equation}
where $A^T = \left[A_i^T\right]_{i=1}^N $ and $\mathrm{N}_{\mathcal{X}_i}$ is the normal cone operator.
Before introducing the algorithms that can be used to solve \ref{GNE}, we discuss some properties of the proposed objective function. It is known that a sufficient condition for the existence and uniqueness of a NE for a n-person non cooperative game is that the system-level objective function $\sigma(x) = \sum_{i=1}^N v_i(x_i)$ is diagonally strictly convex \cite{Rosen1965}. In the case of affine coupling constraints, authors in \cite{Belgioioso2017} and \cite{Paccagnan2016} have shown that the game has a unique variational GNE if the pseudogradient of $\sigma(x)$, $\mathcal{F}:\rm I\!R^{NT} \rightrightarrows \rm I\!R^{NT} = \left[ \partial_{x_i} v_i(x_i) \right]_{i} $ also known as game mapping, is strictly monotone. 
Furthermore, the equilibrium can be reached making the agents pay $A_i\lambda_i$, that is, the value function of each agent coincides with the integral of the first row of KKT in \eqref{GNEKKT}, $\tilde{v_i} = v_i(x_i,x_{-i}) + \lambda^T A_i x_i $. In this case, it has been shown that the agents reach a variational GNE with unique Lagrangian multiplier $\lambda$. 
Note that the game mapping differs from gradient of $\sigma(x)$ since its components are the partial derivatives of the values of the ith agent with respect to its own actions. We now show that the game map generated by the agents' values defined in \eqref{eq:utilities}  inherits monotonicity from the convexity of $e(x)$.

\begin{theorem}
	Let $e(x): \rm I\!R^{NT} \rightarrow \rm I\!R$ be a (strictly/strongly) convex function and let the costs of the agents $c_i(x_i): \rm I\!R^{T} \rightarrow \bar{\rm I\!R}$ be convex functions. Then any repartition $\left[\alpha_i\right]_{i=1}^N$ of $e(x)$ among the agents such that:
	\begin{flalign} \nonumber
		1) \quad &v_i(x_i,x_{-i}) = \alpha_i e(x) + c_i(x_i) &&\\ \nonumber
		2) \quad & \alpha_i \geq 0 \qquad \forall i \in \{N \} 
	\end{flalign} 
	generates a (strictly/strongly) monotone game map $\mathcal{F}: \rm I\!R^{NT} \rightrightarrows \rm I\!R^{NT}$
\end{theorem}
\begin{proof}
$\mathcal{F} = \left[ \partial_{x_i} v_i(x_i,x_{-i})\right]_{i=1}^N$ can be seen as a sum of two operators: $\mathcal{E} =\left[ \partial_{x_i} \alpha_i e(x_i,x_{-i})\right]_{i=1}^N $ and $\mathcal{C} = \left[ \partial_{x_i} c_i(x_i)\right]_{i=1}^N$. Due to the separability of $\mathcal{C}$, it coincides with the gradient of $\sigma(x) = \sum_{i=1}^N c_i(x_i)$. Due to the convexity of $\sigma(x)$, $\mathcal{C}$ is a monotone map, since the gradient of a convex function is monotone (theorem 1 in \cite{Minty1963}).
Using the same reasoning, $\nabla_x e(x)$ is a monotone map due to the convexity of $e(x)$. From the definition of monotonicity, $\left<x-y \vert \nabla_x e(x)-\nabla_y e(y)\right> \geq 0 \quad \forall \quad (x,y)$. Additionally, since any convex function must be convex along any path, we can state it component-wise: 
$(x_i-y_i)(\partial_{x_i}e(x) -\partial_{y_i}e(y))\geq 0 \quad \forall  \quad (i \in \{N\},x,y) $. 
Since we defined all $\alpha_i$ as positive, $(x_i-y_i)\alpha_i(\partial_{x_i}e(x) -\partial_{y_i}e(y))\geq 0 \quad \forall  \quad (i \in \{N\},x,y) $. Thus $\left<x-y \vert \mathcal{E}(x)-\mathcal{E}(y)\right> \geq 0  \quad \forall  \quad (x,y)$, and $\mathcal{F}$ is monotone being the sum of two monotone operators.  
\end{proof}
%
%
%
%
%
%

\section{Algorithms for GNE seeking}\label{s:algo}
As demonstrated in \cite{Paccagnan2016}, asysmmetric projection algorithms \cite{Facchinei2015} can be used to reach a GNE of an aggregative game with quadratic utilities.  Recently, the same algorithm has been rigorously derived modeling the GNE as a monotone inclusion \cite{Belgioioso2018}, showing that it coincides with a preconditioned forward backward (pFB) method (algorithm \eqref{alg:1}), which is a special case of the Banach-Picard iteration \cite{Bauschke2011} of two operators whose sum is the set value mapping associated to the KKT conditions in \eqref{GNEKKT}.
\begin{algorithm}
	\caption{pFB}\label{alg:1}
	\begin{algorithmic}
		\State $x^{k+1} = \boldsymbol{\Pi}_{\mathcal{X}} \left[x^{k} -\alpha (\mathcal{F}(x^{k})+A^t\lambda^k)) \right]$ 
		\State $\lambda^{k+1} = \boldsymbol{\Pi}_{\rm I\!R^+} \left[\lambda^{k} +\beta (2Ax^{k+1}-Ax^{k}-b) \right]$ 
	\end{algorithmic}
\end{algorithm}

We compare algorithm \eqref{alg:1} with a trivial modification of the ADMM algorithm \cite{Boyd2010}, which convergence rate and properties have been extensively studied in the literature.
For clarity of exposition, we start considering the version of problem \eqref{eq:sharing} without coupling constraints. This can be solved in a centralized way through ADMM, applying the procedure in \cite{Boyd2010}  \S 7.3, which results in the following parallelized formulation:

\begin{algorithm}
	\caption{ADMM}\label{alg:2}
	\begin{algorithmic}
    \State  
      \begin{align*}
    	 x_i^{k+1} &= \argmin{x_i \in \mathcal{X}_i} c_i(x_i) +\frac{\alpha_i}{2\rho} \Vert (S x^k -y^k)/N  \\
    	&-x_i^k +x_i +\lambda^k \Vert_2^2\\
    	& +\frac{1}{2\rho} \Vert (A x^k -y^k)/N -A_ix_i^k +x_i +\lambda_a^k \Vert_2^2 \label{eq:agent_min}\numberthis \\
    	y^{k+1} &= \argmin{y} e(y) + \frac{1}{2\rho} \Vert y -S x^{k+1}  -\lambda^k \Vert \numberthis \label{eq:center_min} \\
    	\lambda^{k+1} &= \lambda^{k} +S x^{k+1} -y^{k+1}\label{eq:lambda_center} \numberthis \\
    	y_a^{k+1} &=  \argmin{y} \mathcal{I}_{\mathcal{X}_a} + \frac{1}{2\rho} \Vert y_a -Ax^{k+1}  -\lambda^k \Vert \numberthis \label{eq:y_a_min}  \\
    	\lambda_a^{k+1} &= \lambda_a^{k} +A x^{k+1} -y_a^{k+1}\label{eq:lambda_a} \numberthis 
    \end{align*}
	\end{algorithmic}
\end{algorithm}


where the only difference form the centralized algorithm is the $\alpha_i$ coefficient in the $x_i$ update.
We can write the KKT conditions at convergence

\begin{subnumcases}
\partial \partial_{x_i} c_i(x_i^*) + \alpha_i \frac{\lambda^*}{\rho} +A_i^T\frac{\lambda_a^*}{\rho} +\mathrm{N}_{\mathcal{X}_i} = 0 \qquad \forall i \in N &\label{KKT1} \\ 
\partial_{y}  e(y^*) -\frac{\lambda^*}{\rho} = 0& \label{KKT2} \\ 
y^* = Sx^*& \label{KKT3}\\
0 \leq \lambda_a^* \perp -(Ax-b) \geq 0 \label{KKT4} 
\end{subnumcases} 
We can find $\lambda^*$ from \ref{KKT2} and substitute it in \ref{KKT1}:
\begin{equation}
\partial_{x_i} c_i(x_i^*) + \alpha_i \partial_{y}  e(y^*)  +A_i^T\frac{\lambda_a^*}{\rho} + \mathrm{N}_{\mathcal{X}_i}= 0
\end{equation}
then we can use \ref{KKT3}, and recalling that $S$ is the summation matrix, we obtain:
\begin{equation}
\partial_{x_i} c_i(x_i^*) + \alpha_i \partial_{x_i}  e(x^*) +A_i^T\frac{\lambda_a^*}{\rho} + \mathrm{N}_{\mathcal{X}_i} = 0
\end{equation}
which, together with \ref{KKT4} are equivalent to the KKT \ref{GNEKKT} of the game \ref{GNE}, when $v_i = c_i(x_i) + \alpha_i e(x)$.

\subsection{Pricing and individual rationality} 
In this paper we only consider the case in which the function $e(x)$ is the surplus that the agent community has in paying the energy at the point of common coupling with the electrical grid:
\begin{equation}\label{eq:surplus}
e(x) = c\left(\sum_{i=1}^{N} x_i\right)-\sum_{i=1}^{N} c(x_i) 
\end{equation}
where $x_i \in \rm I\!R^{T}$ is the vector of total power of the ith agent, $c(\cdot)$ is the energy cost function defined as:
\begin{equation}\label{eq:cost_fun}
c(z_t) = 
\begin{cases}
p_{b,t} z_t , & \text{if } \quad z_t \geq 0  \\
p_{s,t} z_t ,              & \text{otherwise}
\end{cases} 
\end{equation}
where $p_{b,t}$ and $p_{s,t}$ are the buying and selling tariffs, respectively, at time $t$. In order to induce agents to follow the proposed mechanism, we must ensure that the energy tariff they pay participating in the market is always lower than the one they pay in the base case. This is always true when we are not taking into account grid constraints, since $e(x)$ as defined in \eqref{eq:surplus} is always non-negative, when $p_{b,t} \geq p_{s,t}$, as usual in energy tariffs. However, if the agents are located in a grid with big voltage oscillations, the Lagrangian dual variables (which we can interpret as punishment prices) could be such that the cost paid by the agents is higher than $\alpha_i e(x)$. To ensure IR, we encode it in the optimization scheme. At each iteration, for each time step in the horizon, we increment the Lagrangians only if the following condition holds:
\begin{equation} \label{the_condition}
\alpha_i e(x^k_t) + A_i^T\lambda^k_t \leq 0 \qquad \forall i \in N, \quad \forall t \in T
\end{equation}
where a negative value means that the prosumer is gaining a reward. This obviously results in the impossibility to satisfy the coupling constraints. We can give the following straightforward economic interpretation to this mechanism: each agent would opt-out from the game as soon as the energy tariffs become unfavorable with respect to the existing one. Condition \eqref{the_condition} prevent this from happening. In the presence of bad power quality, the DSO could provide favorable energy tariffs to prosumers participating in the mechanism, ensuring that condition \eqref{the_condition} is met with high probability.

\subsection{Prosumers problem formulation}\label{problem_form} 
In this paper, each prosumer's flexibility is modeled using an electric battery. 
Although simple, the model we used is not simplistic and we briefly describe it in this subsection.   Since the effect of charging or discharging on the state of charge is not symmetric due to the efficiencies, the problem is usually formulated as a mixed integer linear program (MILP), introducing binary decision variables and using bilinear constraints, to avoid the simultaneous charge and discharge of the battery. Furthermore, the objective function of the agents is non differentiable at $P_m = 0$ and is mathematically described by the maximum operator. In order to speed up the computations, we reformulated all the control problems as a quadratic optimization. 

We start considering that both the ADMM and the pFB formulations can be described by the following optimization problem:
\begin{equation}\label{eq:prox_prob}
\argmin{x_i \in \mathcal{X}_i}  \ f(x_i,x_{-i})  +\frac{1}{2\rho}\Vert D x_i -r^k \Vert_2^2 
\end{equation}

where $r$ is a reference signal and $D \in \rm I\!R^{T \times 2T} = I_{T} \otimes [1,-1] $, performs the sum of the charging and discharging operations with appropriate signs. Here, with abuse of notation, we redefined the vector $x_i\in \rm I\!R^{2T}$ as the vector containing both the charging and discharging operators, in such a way that $x_i = [P_{in,t};P_{out,t}]_{t=1}^T$, where $P_{in,t}$ and $P_{out,t}$ are the charging and discharging powers of the battery. 
For the ADMM formulation, it is easy to see that problem \eqref{eq:prox_prob} can be used to solve \eqref{eq:agent_min}. We can still use  \eqref{eq:prox_prob} for solving the pFB formulation recalling that the projected gradient descent is equivalent to a quadratic optimization problem in the form:
\begin{equation}\label{eq:asy_prob}
\argmin{x_i \in \mathcal{X}_i}  \ \left(\mathcal{F}_i(x_i^{k})+A_i^t\lambda^k\right)^T x_i  +\frac{1}{2\rho}\Vert x_i -x_i^k \Vert_2^2 
\end{equation}

The battery is modeled as a discrete linear system, with the state of charge denoted by $s$:
\begin{equation}
s_{i,t+1} = A_{s,i} s_{i,t} + B_{s,i} x_{i,t}
\end{equation} 
We can eliminate the dependence on the state of the optimization problem, using the standard batch formulation:
\begin{equation}
s_i = \Lambda s_{i,0} + \Gamma x_i
\end{equation}
where $x_i \in \rm I\!R^{2T \times 1}$ is the control vector for the whole time horizon $T$ and $\Lambda \in \rm I\!R^{T \times 1}, \Gamma \in \rm I\!R^{T \times 2T}$ are the batch matrices.   

We can now describe the set $\mathcal{X}_i$ through the linear constraints $A_{c_i} x_i \leq b_{c_i}$, defined as:  
\begin{equation}
A_{c_i} = \begin{bmatrix} I \\ I \\  -\Gamma \\ \Gamma \end{bmatrix} 
b_{c_i} =\begin{bmatrix} x_{min} \\ x_{max} \\  -e_{min} + \Lambda e_0 \\ e_{max} -\Lambda e_0 \end{bmatrix}
\end{equation}
where $x_{min}, x_{max}, e_{min}, e_{max}$ $ \in \rm I\!R^{2T} $ are the power and energy box constraints, while $I$ is the identity matrix of appropriate dimensions.
Now we can reformulate the non differentiable cost function \eqref{eq:cost_fun} with a linear function, such as we can reuse it in both the ADMM and pFB formulations. We start considering that the if condition of the cost function in \eqref{eq:cost_fun} can be equivalently formulated using the max operator. In turn, the max operator can be replaced by the sum of an auxiliary variable $y$ and appropriate inequality constraints. We augment our decision variable such as $\tilde{x} = [x^T,y^T]^T$. Now the minimization of \eqref{eq:cost_fun} is equal to the following optimization problem:
\begin{equation}\label{eq:selfish_problem_2}
\begin{aligned}
\min_{\tilde{x}} & \ l^T \tilde{x} \\
s.t.:\ &  \tilde{A} \tilde{x}  \leq \tilde{b}\\
\end{aligned}
\end{equation}

where $\tilde{A} = [A_{c,i};A_y]$ and $\tilde{b} = [b_{c,i},b_y]$, and 
\begin{equation}\label{eq:aug_matrices}
A_y = \begin{bmatrix} D \circ P_b - I_T \\ D \circ P_s  -I_T\end{bmatrix} \quad
b_y =\begin{bmatrix} -p_b P_m \\ -p_s P_m  \end{bmatrix}
\end{equation}

where the  $P_b, P_s \in \rm I\!R^{T \times 2T} \ $  $t_{th}$ rows entries are identical to the buying and selling prices at time $t$.
The effect of the matrices in \eqref{eq:aug_matrices} is that the new auxiliary variable $y$ is now an upper envelope for the cost function \eqref{eq:cost_fun}. Since we require the cost to be minimized, $y$ will coincide with the cost function $c(\cdot)$ at optimality. We can now use $A_y$ and $b_y$ in both the ADMM and pFB formulations. While in the first $l^T x_i$ replaces $f(x_i,x_{-i})$ in \eqref{eq:prox_prob}, in the latter the prosumers' energy costs are considered as part of the pseudogradient: $F_i = l + \partial_{x_i} e(x)$.

This problem formulation prevent us from introducing binary variables for the charging and discharging powers, since optimal solutions of \eqref{eq:selfish_problem_2} does not require to simultaneously charge and discharge the battery at the same time. 
This is not true for the ADMM formulation, in which the quadratic penalty on the sum of $P_{in,t}$ and $P_{out,t}$ with respect to a reference signal $r^k$ is present. In the case in which the reference signal $r^k$ is negative, the battery is not only incentivized to charge itself, but to consume as much energy as possible. This will result in a simultaneous charge and discharge, due to the round-trip efficiency. To avoid this behavior, $f(\cdot)$ can be augmented with a linear therm, punishing the battery discharging operation when $r$ is negative: $\tilde{f}(\tilde{x}_i) = f(\tilde{x}_i) + l_p^T \tilde{x}_i $ where $l_p  \in \rm I\!R^{1 \times 3*T}$ has non zero entries, all identical to a punishment therms, only when $r<0$.


\section{Numerical analysis}\label{s:ana}
We test both the ADMM and the pFB and compare the performance to a centralized solution. The only difference from the ADMM and the centralized formulation are the $\alpha_i$ coefficients in equation \eqref{eq:agent_min}, which are not present in the centralized solution. Since the system-level objective function $e(x)$, as defined in \eqref{eq:surplus}, is not differentiable in 0 and is not strictly nor strongly convex, the convergence of pFB is not guaranteed. To have an equal comparison, we replaced the system-level cost function \eqref{eq:cost_fun} with a continuously differentiable function. We define it by means of its derivative:
\begin{equation}
\nabla_z \tilde{c}(z) = (p_{b,t}-p_{s,t})\frac{ \tanh(k Sx_{t})+1}{2} + p_{s,t}
\end{equation}
where $k$ regulates the steepness of the function in $z=0$. 
\begin{figure}[h]
	\centering
	\includegraphics[width=1\linewidth]{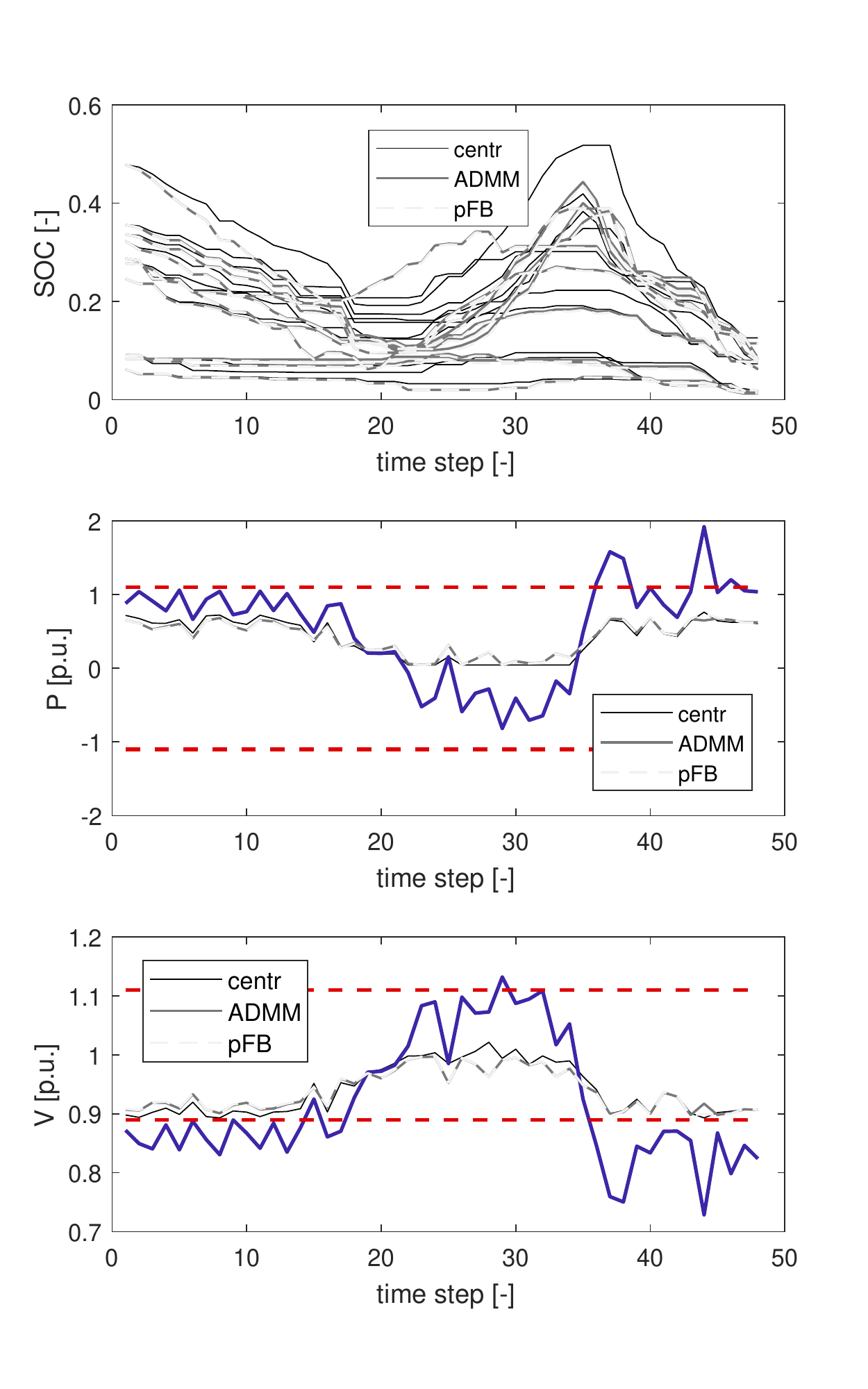}
	\caption{Time series example, N = 10. Blue: forecasted profiles. Red: constraints. Grays: solutions of the centralized and decentralized approaches. Top: state of charge for each battery. Middle: power profiles. Bottom: voltage profiles.}
	\label{fig:timeseries}
\end{figure}
In our simulations $k=10$, which provides a reasonable steepness for all the possible values of the power aggregate, since we did all the computations in per units, and the aggregate power constraint is $Sx \in \left[-1.1, 1.1\right]$. We stress out that this approximation is only used for the system-level objective, and not for the prosumers objective functions, where the cost \eqref{eq:cost_fun} is modeled as described in subsection \ref{problem_form}. In order to fairly compare the algorithms, we used an equal stepsize $\rho$, fixed to 0.1.
The power profiles of each prosumer are randomly chosen from a yearly dataset of real residential electrical consumption. Each prosumer is equipped with a PV field, with a nominal power uniformly distributed between 2 and 10 times its daily energy consumption. Furthermore, each prosumer is  provided with an electric battery with size equal to the expected daily energy exceeding its consumption. Figure \ref{fig:timeseries} shows the optimized time series from a single case. In the upper panel, the batteries' state of charge (SOC) are shown. Since the SOC is the time integral of the optimization variables ($P_{in}, P_{out}$), it is clear that the ADMM and the pFB converged exactly to the same solution. The middle panels shows the forecasted aggregated power profile and the optimized one. Note that both the ADMM and pFB solutions are not far from the centralized solution, while differences are more evident in terms of single prosumers SOC. The last panel shows voltage profiles at the point of common coupling. In figure \ref{fig:convergence} the convergence of the two algorithms is shown, in terms of game objective function $\sigma(x)$. We ran a total of 50 simulations, each of which includes 10 prosumers with power profiles and battery sizes randomly chosen, as explained before. For each simulation $s$, we retrieve the best optimal value of $\sigma(x)$, $p_{best}^s$, defined as:
\begin{equation}
p_{best}^s = \text{minimum} \ \{p^{*s}_{ADMM},p^{*s}_{pFB}\} 
\end{equation} 
where $p^{*s}_{ADMM},p^{*s}_{pFB}$ are the solution of the two algorithms after 200 iterations (after which the relative change in $\sigma(x)$ for all the simulations was smaller than $1e-5$). The thick lines show the median, while the shadowed patches contain half of the simulations.

\begin{figure}[h]
	\centering
	\includegraphics[width=1\linewidth]{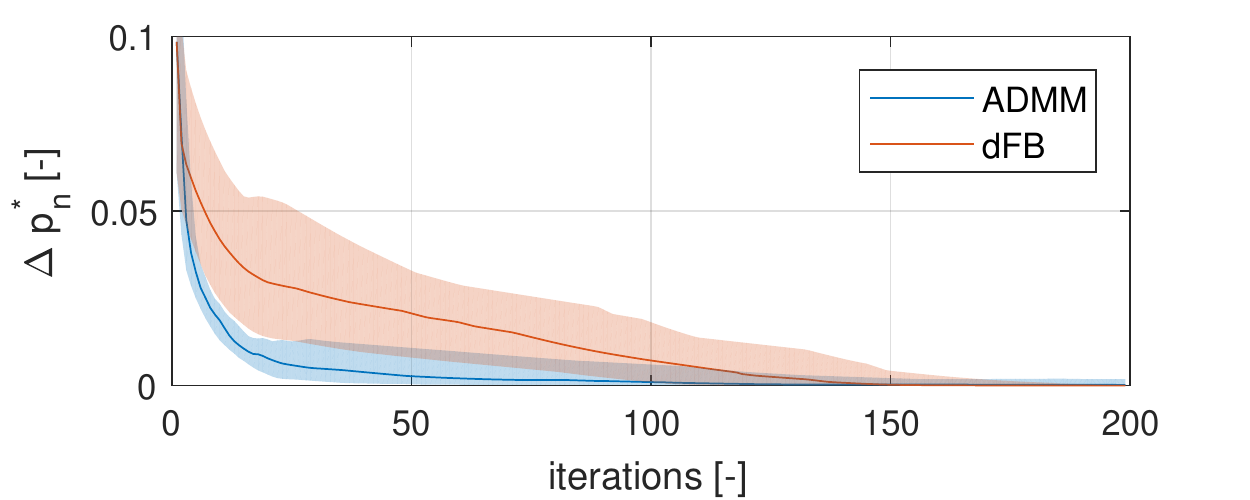}
	\caption{Normalized optimal value $p^*$. The thick lines denote the median, while the shaded areas are the $25\%$ and $75\%$ quantiles. }
	\label{fig:convergence}
\end{figure}

\section{Conclusions}
We have proposed a method to enforce IR while reaching a EPGNE in a distributed way. The method and the related algorithm have been tested, and compared with pFB, a state of the art algorithm for GNE seeking. The simulations shows that the proposed algorithm reaches the same solutions of pFB, while showing faster convergence in most of the cases. In future research, we will extensively investigate the advantages of the proposed methodology.

\bibliographystyle{IEEEtran}

\begin{thebibliography}{10}
	\providecommand{\url}[1]{#1}
	\csname url@samestyle\endcsname
	\providecommand{\newblock}{\relax}
	\providecommand{\bibinfo}[2]{#2}
	\providecommand{\BIBentrySTDinterwordspacing}{\spaceskip=0pt\relax}
	\providecommand{\BIBentryALTinterwordstretchfactor}{4}
	\providecommand{\BIBentryALTinterwordspacing}{\spaceskip=\fontdimen2\font plus
		\BIBentryALTinterwordstretchfactor\fontdimen3\font minus
		\fontdimen4\font\relax}
	\providecommand{\BIBforeignlanguage}[2]{{%
			\expandafter\ifx\csname l@#1\endcsname\relax
			\typeout{** WARNING: IEEEtran.bst: No hyphenation pattern has been}%
			\typeout{** loaded for the language `#1'. Using the pattern for}%
			\typeout{** the default language instead.}%
			\else
			\language=\csname l@#1\endcsname
			\fi
			#2}}
	\providecommand{\BIBdecl}{\relax}
	\BIBdecl
	
	\bibitem{Molzahn2017}
	D.~K. Molzahn, F.~Dorfler, H.~Sandberg, S.~H. Low, S.~Chakrabarti, R.~Baldick,
	and J.~Lavaei, ``{A Survey of Distributed Optimization and Control Algorithms
		for Electric Power Systems},'' \emph{IEEE Transactions on Smart Grid}, vol.
	3053, no.~c, pp. 1--1, 2017.
	
	\bibitem{Makowski1987}
	L.~Makowski and J.~M. Ostroy, ``{Vickrey-Clarke-Groves mechanisms and perfect
		competition},'' \emph{Journal of Economic Theory}, vol.~42, no.~2, pp.
	244--261, 1987.
	
	\bibitem{Clarke1971}
	E.~H. Clarke, ``{Multipart pricing of public goods},'' \emph{Public Choice},
	vol.~11, no.~1, pp. 17--33, 1971.
	
	\bibitem{Vickrey1961}
	W.~Vickrey, ``{Counterspeculation, auctions, and competitive sealed tenders.},''
	\emph{The Journal of Finance}, vol.~16, no.~1, pp. 8--37, 1961.
	
	\bibitem{Poolla2017}
	\BIBentryALTinterwordspacing
	B.~K. Poolla, S.~Bolognani, L.~Na, and F.~Dorfler, ``{A Market Mechanism for
		Virtual Inertia},'' \emph{arXiv}, 2017.
	\BIBentrySTDinterwordspacing
	
	\bibitem{Conitzer2006}
	\BIBentryALTinterwordspacing
	V.~Conitzer and T.~Sandholm, ``{Failures of the VCG mechanism in combinatorial
		auctions and exchanges},'' \emph{The 5th International Joint Conference on
		Autonomous and Multiagent Systems}, pp. 521--528, 2006.
	\BIBentrySTDinterwordspacing
	
	\bibitem{Rothkopf2007}
	\BIBentryALTinterwordspacing
	M.~H. Rothkopf, ``{Thirteen Reasons Why the Vickrey-Clarke-Groves Process Is
		Not Practical},'' \emph{Operations Research}, vol.~55, no.~2, pp. 191--197,
	2007.
	\BIBentrySTDinterwordspacing
	
	\bibitem{Tardos2007}
	N.~Noam, T.~Roughgarden, E.~Tardos, and T.~Wexler, \emph{{Algorithmic Game
			Theory}}.\hskip 1em plus 0.5em minus 0.4em\relax Cambridge University Press,
	2007.
	
	\bibitem{Petcu2008}
	A.~Petcu, B.~Faltings, and D.~C. Parkes, ``{M-DPOP: Faithful distributed
		implementation of efficient social choice problems},'' \emph{Journal of
		Artificial Intelligence Research}, vol.~32, pp. 705--755, 2008.
	
	\bibitem{Narahari2014}
	\BIBentryALTinterwordspacing
	Y.~Narahari, \emph{{Game Theory and Mechanism Design}}, 2014, vol.~4.
	\BIBentrySTDinterwordspacing
	
	\bibitem{Parkes2004}
	\BIBentryALTinterwordspacing
	D.~C. Parkes and J.~Shneidman, ``{Distributed Implementations of
		Vickrey-Clarke-Groves Mechanisms},'' \emph{Third Joint Conference on
		Autonomous and Multiagent Systems}, pp. 261--268, 2004. 
	\BIBentrySTDinterwordspacing
	
	\bibitem{Tanaka2018}
	T.~Tanaka, F.~Farokhi, and C.~Langbort, ``{Faithful implementations of
		distributed algorithms and control laws},'' \emph{IEEE Transactions on
		Control of Network Systems}, vol.~4, no.~2, pp. 191--201, 2017.
	
	\bibitem{Rosen1965}
	J.~. B.~. Rosen, ``{Existence and Uniqueness of Equilibrium Points for Concave
		N-Person Games},'' \emph{Econometrica}, vol.~33, no.~3, pp. 520--534, 1965.
	
	\bibitem{Kim2013a}
	B.-g. Kim, S.~Member, S.~Ren, M.~V.~D. Schaar, J.-w. Lee, and S.~Member,
	``{Bidirectional Energy Trading and Residential Load Scheduling with Electric
		Vehicles in the Smart Grid},'' \emph{IEEE Journal on selected areas in
		communication}, vol.~31, no.~7, pp. 1219--1234, 2013.
	
	\bibitem{Gharesifard2016}
	B.~Gharesifard, T.~Başar, and A.~D. Dom{\'{i}}nguez-Garc{\'{i}}a,
	``{Price-based coordinated aggregation of networked distributed energy
		resources},'' \emph{IEEE Transactions on Automatic Control}, vol.~61, no.~10,
	pp. 2936--2946, 2016.
	
	\bibitem{Yang2010}
	B.~Yang and M.~Johansson, ``{Distributed optimization and games: A tutorial
		overview},'' in \emph{Networked Control Systems}, 2010, vol. 406, pp.
	109--148.
	
	\bibitem{Belgioioso2018}
	\BIBentryALTinterwordspacing
	G.~Belgioioso and S.~Grammatico, ``{Projected-gradient algorithms for
		generalized equilibrium seeking in Aggregative Games are preconditioned
		Forward-Backward methods},'' 2018. [Online]. Available:
	\url{http://arxiv.org/abs/1803.10441}
	\BIBentrySTDinterwordspacing
	
	\bibitem{Almasalma2017}
	H.~Almasalma, J.~Engels, and G.~Deconinck, ``{Dual-decomposition-based
		peer-to-peer voltage control for distribution networks},'' \emph{CIRED - Open
		Access Proceedings Journal}, no.~1, pp. 1718--1721, 2017.
	
	\bibitem{Mugnier2016}
	C.~Mugnier, K.~Christakou, J.~Jaton, M.~{De Vivo}, M.~Carpita, and M.~Paolone,
	``{Model-less/measurement-based computation of voltage sensitivities in
		unbalanced electrical distribution networks},'' \emph{19th Power Systems
		Computation Conference, PSCC 2016}, 2016.
	
	\bibitem{Weckx2015}
	S.~Weckx, S.~Member, R.~D. Hulst, J.~Driesen, and S.~Member, ``{Voltage
		Sensitivity Analysis of a Laboratory Distribution Grid With Incomplete
		Data},'' \emph{IEEE Transactions on Smart Grid}, vol.~6, no.~3, pp.
	1271--1280, 2015.
	
	\bibitem{Jensen2010}
	M.~K. Jensen, ``{Aggregative games and best-reply potentials},'' \emph{Economic
		Theory}, vol.~43, no.~1, pp. 45--66, 2010.
	
	\bibitem{Belgioioso2017}
	G.~Belgioioso and S.~Grammatico, ``{Semi-Decentralized Nash Equilibrium Seeking
		in Aggregative Games With Separable Coupling Constraints and
		Non-Differentiable Cost Functions},'' \emph{IEEE Control Systems Letters},
	vol.~1, no.~2, pp. 400--405, 2017.
	
	\bibitem{Paccagnan2016}
	D.~Paccagnan, M.~Kamgarpour, B.~Gentile, F.~Parise, J.~Lygeros, and
	D.~Paccagnan, ``{Distributed computation of Nash Equilibria in aggregative
		games with coupling constraints},'' in \emph{2016 IEEE 55th Conference on
		Decision and Control (CDC)}, Las Vegas, USA, 2016, pp. 6123--6128.
	
	\bibitem{Minty1963}
	G.~J. Minty, ``{On the Monotonicity of the Gradient of a Convex Function},''
	\emph{Pacific J. Math.}, vol.~14, no.~1, pp. 243--248, 1963.
	
	\bibitem{Facchinei2015}
	F.~Facchinei and J.~S. Pang, \emph{{Finite-Dimensional Variational Inequalities
			and Complementarity Problems}}, 2015, vol. 1542, no.~9.
	
	\bibitem{Bauschke2011}
	H.~H. Bauschke and P.~L. Combettes, \emph{{Convex Analysis and Monotone
			Operator Theory in Hilbert Spaces}}, 2011.
	
	\bibitem{Boyd2010}
	\BIBentryALTinterwordspacing
	S.~Boyd, N.~Parikh, E.~Chu, B.~Peleato, and J.~Eckstein, ``{Distributed
		Optimization and Statistical Learning via the Alternating Direction Method of
		Multipliers},'' \emph{Foundations and Trends{\textregistered} in Machine
		Learning}, vol.~3, no.~1, pp. 1--122, 2010. 
	\BIBentrySTDinterwordspacing
	
\end{thebibliography}

\end{document}